\documentclass[12pt]{amsart}

\usepackage{tgpagella}
\usepackage{euler}
\usepackage[T1]{fontenc}
\usepackage{amsmath, amssymb}
\usepackage[hidelinks]{hyperref}
\usepackage[english]{babel}
\usepackage{mathrsfs}
\usepackage{eucal}
\usepackage[all]{xy}
\usepackage{tikz}

\newtheorem{thm}{Theorem}
\newtheorem*{thm*}{Theorem}
\newtheorem{lem}[thm]{Lemma}
\newtheorem{fact}[thm]{Fact}

\newtheorem*{prop*}{Proposition}

\newtheorem*{cor*}{Corollary}

\theoremstyle{definition}

\newtheorem*{defn*}{Definition}

\newtheorem{remark}[thm]{Remark}

\newtheorem*{question*}{Question}
\newtheorem*{Pquestion*}{Popa's question}

\newtheorem*{conv*}{Convention}

\def\bb{\mathbb}

\def\bb{\mathbb}

\makeatletter

\def\dotminussym#1#2{%
  \setbox0=\hbox{$\m@th#1-$}%
  \kern.5\wd0%
  \hbox to 0pt{\hss\hbox{$\m@th#1-$}\hss}%
  \raise.6\ht0\hbox to 0pt{\hss$\m@th#1.$\hss}%
  \kern.5\wd0}

\def \st{\operatorname{st}}

\def \starR{{}^*\mathbb R}
\def \starA{{}^*A}
\def \starB{{}^*B}
\def \starL{{}^*L}
\def \starC{{}^*\mathbb C}

\def \starpsi{{}^*\psi}

\begin{document}

%%%%%%%%%%%%%%%%%%%%%%%%%%%%%%%%%%%%%%%%%%%%%%

\title{A Nonstandard Formulation of Bohmian Mechanics}
\author{Jeffrey Barrett and Isaac Goldbring}
\thanks{Goldbring was partially supported by NSF grant DMS-2054477.}

\address{Department of Logic and Philosophy of Science\\University of California, Irvine, 765 Social Sciences Tower,
Irvine, CA 92697-5000}
\address{Department of Mathematics\\University of California, Irvine, 340 Rowland Hall (Bldg.\# 400),
Irvine, CA 92697-3875}
\email{j.barrett@uci.edu}
\urladdr{
https://faculty.sites.uci.edu/jeffreybarrett/}
\email{isaac@math.uci.edu}
\urladdr{http://www.math.uci.edu/~isaac}

\maketitle

\begin{abstract}
Using the tools of nonstandard analysis, we develop and present an alternative formulation of Bohmian mechanics. This approach allows one to describe a broader assortment of physical systems than the standard formulation of the theory. It also allows one to make predictions in more situations. We motivate the nonstandard formulation with a Bohmian example system that exhibits behavior akin to Earman's (1986) classical space invaders and reverse space invaders. We then use the example to illustrate how the alternative formulation of Bohmian mechanics works.
\end{abstract}

Bohmian mechanics is the best known and most thoroughly studied hidden-variable formulation of quantum mechanics. The theory was first presented by David Bohm (1952), then subsequently sharpened by John Bell (1987). We will begin with a brief discussion of the theory. We will then describe a physical situation for which the theory predicts a strong failure of determinism akin to John Earman's (1986) space invaders in classical mechanics. Finally, we will show how a nonstandard formulation of Bohmian mechanics allows one to describe such physical situations and to make clear predictions regarding how the state evolves.

\section{Bohmian mechanics}

Following Bell's (1987) formulation of the theory, Bohmian mechanics can be characterized by four rules:

\begin{itemize}
    
\item [1.] \emph{Representation of states}: The complete physical state of a system~$S$ at time~$t$ is given by the wave function $\psi(q,t)$ over configuration space and a point in configuration space~$Q(t)$.

\item [2.] \emph{Interpretation of states}: The position of every particle is always determinate and is given by the current configuration~$Q(t)$.

\item [3I.] \emph{Linear dynamics}: The wave function evolves in the standard unitary way
$$
i \hbar \frac{\partial \psi(q,t)}{\partial t} =  \hat{H} \psi(q,t)
$$
where $\hat{H}$ is the Hamiltonian.

\item [3II.]  \emph{Particle dynamics}: Particles move according to
$$
\frac{d Q_k(t)}{dt} = \frac{1}{m_k} \frac{\mbox{Im }\psi^*(q,t) \nabla_k \psi(q,t)}{\psi^*(q,t)\psi(q,t)}
\Big|_{Q(t)}
$$
where~$m_k$ is the mass of particle~$k$ and~$Q(t)$ is the current particle configuration.

\item [4.]  \emph{Distribution postulate}: The probability density of the configuration~$Q(t_0)$ is $|\psi(q, t_0)|^2$ at an initial time~$t_0$.

\end{itemize}

Both the wave function $\psi(q,t)$ and the particle configuration $Q(t)$ evolve in $3N$-dimensional configuration space, where $N$ is the number of particles in the system one wishes to describe. The $3N$-coordinates of the configuration $Q(t)$ give the position of each particle at time $t$. One can think of the probability density $|\psi(q, t)|^2$ as describing the density of a compressible fluid in configuration space. The wave function evolves according to the linear dynamics (rule~3I), and as the compressible fluid flows about in configuration space, it carries the point representing the particle configuration $Q(t)$ as described by the particle dynamics (rule~3II). As a result, the configuration moves in configuration space as if it were a massless particle carried by the probability current.

In contrast with collapse theories like the standard von Neumann (1932) theory and GRW (1986), the dynamics in Bohmian mechanics is both deterministic and time-reversal symmetric for a broad range of simple physical systems. If one knows the wave function at time $t$ and the Hamiltonian $\hat{H}$, then rule~3I determines the wave function at all future and past times. And in many situations if one knows how the wave function evolves and the particle configuration $Q(t)$ at a time, then rule~3II determines the positions of the particles at all future and past times. But there are some situations where the theory is less well behaved. Before considering what can go wrong, it is important to understand how probability works in the theory.

Quantum probabilities are purely epistemic in Bohmian mechanics. They result from a physical observer not knowing the initial particle configuration. The dynamics has the property that if the epistemic probability density for the particle configuration is ever given by the standard epistemic quantum probabilities $|\psi(q, t)|^2$, then it will continue to be until one makes a measurement. After a measurement, it will be given by the effective wave function, a notion introduced and discussed by D\"urr, Goldstein, and Zangh\`i (1992).\footnote{One gets empirical predictions by supposing that measurement records supervene on the effective wave function. See D\"urr, Goldstein, and Zangh\`i (1992) and Barrett (2000), (2020), (2021) for discussions of how this works.} The distribution postulate (rule~4) stipulates that the epistemic probability density for the particle configuration at time~$t_0$ is $|\psi(q, t_0)|^2$. It is this statistical boundary condition together with how the dynamics works that yields the standard quantum probabilities as epistemic probabilities over particle configurations.

Rule~4 is crucial to the empirical adequacy of the theory. If particles were not distributed in this way, the theory would not predict the standard quantum probabilities. And if an observer ever knew the particle configuration with more precision than allowed by the standard quantum probabilities, she would be able to predict the results of her future measurements more precisely than allowed by the standard quantum probabilities.

Some proponents of Bohmian mechanics do not like having to assume a special boundary condition like rule~4 as a part of the theory. As a result, there is a long tradition of seeking to derive something like rule~4 from the Bohmian dynamics and some collection of general epistemic principles.\footnote{D\"urr, Goldstein, and Zangh\`i (1992) and their subsequent work on justifying the assumption of an equivariant epistemic distribution is a salient example of this tradition. Such approaches often appeal to a version of Cournot’s principle to argue from the fact that some property is typically true relative to a specified measure of typicality to the conclusion that the property can be assumed to hold. See Diaconis and Skyrms (2018, 66-7) for a brief discussion of Cournot’s principle.} On this approach, one treats the wave function and configuration as independent then argues that one can expect rule~4 to be eventually satisfied under the dynamics.

Allowing the wave function and particle configuration to be independent provides a broader range of physical possibilities, but it also allows for physical situations where the behavior of a system is radically underdetermined by its state. We will briefly discuss determinism in classical mechanics then consider what can happen in a Bohmian system when the wave function and particle position are logically independent.

\section{Classical invaders}

In his discussion of determinism in classical mechanics, John Earman showed how a particle might move beyond every spatial location in a finite time and how a particle might move to any specified spatial region without having been at any spatial location a finite time earlier (1986, 34--5). Following Mather and MaGehee (1975), Earman further showed how each phenomenon might occur as the result of physically possible forces (1986, 35--7). A simple example suffices to illustrate the idea.

Consider a one-dimensional universe containing a single particle that starts at a location $x_0$, as specified by an inertial observer, then moves one meter to the right in $1/2$ second, another meter to the right in the next $1/4$ second, another meter to the right in the next $1/8$ second, and so on. After $1$ second the particle will be beyond any specified location to the right and hence not spatially located at all. This phenomenon is sometimes called a \emph{reverse space invader}. And inasmuch as classical mechanics is time-reversal symmetric, this process is reversible. In that case, a particle that is not at any spatial location one second ago moves to location $x_0$, the spatial part of its trajectory mirroring that of a reverse space invader. This is a \emph{space invader}. Reverse space invaders pose a problem for principles like the conservation of mass, energy, momentum, and charge; space invaders pose a direct problem for determinism.

Classical mechanics allows that a space invading particle might appear come flying in at any time disturbing the predicted behavior of an otherwise deterministic system. While one might seek to recover determinism by imposing constraints on classical mechanics that would eliminate such phenomena, Earman showed that finding plausible physical constraints that do so is more difficult than one might at first imagine (1986, 37--53). He concluded:
\begin{quote}
Newtonian space-time, whose structure is rich enough to support the possibility of determinism in classical worlds, nevertheless proves to be a none too friendly environment. The principle irritant derives from the possibility of arbitrary fast causal signals, threatening to trivialize domains of dependence. (1986, 52)
\end{quote}
As we shall see, at least part of the problem regarding domains of dependence results from classical mechanics's use of the reals $\mathbb{R}$ to model the spacetime continuum. While we will focus here on Bohmian mechanics, on moral of the story is that the hyperreals ${}^*\mathbb R$ provide a natural model for the continuum that is particularly well suited to handling classical invader-type phenomena.\footnote{See Wontner (2019) for an argument that the hyperreals may prove a better candidate than the reals for representing continuous phenomena more generally. He argues, in particular, that they more naturally capture the ``gap-free'' nature of continuity.}

\section{Bohmian invaders}
\label{invader}

Invader-type phenomena occur in Bohmian mechanics in simple physical situations. Consider a one-dimensional universe containing a single particle and no fields. Since there is only one particle, the configuration space is one-dimensional position space. Suppose at time $t_0$ the wave function $\psi(x,t_0)$ begins entirely within the interval $[0,1]$
$$
\int_0^1 |\psi(x,t_0)|^2 dx =1
$$
and the particle begins at position $x=2$. By rule~3I, for any time $t>t_0$
$$
\int_k^\infty |\psi(x,t)|^2 dx > 0
$$
for any $k$. This results from the fact that $\psi(x,t_0)$ has positive support only in the interval $[0,1]$ in position space and hence its Fourier transform is nonzero for arbitrarily large momenta.

By rule~3II, the particle will be carried by the probability current to the right. Since the particle's motion is perfectly sensitive to the probability flow (see Fact 4 below), the integral of $|\psi(x, t)|^2$ to the right of the particle's position $x(t)$ must always be zero. This requires that the particle be to the right of every finite position which means that it is no longer at any spatial location.

This also means that if one insists that the dynamics is deterministic and hence time-reversal symmetric, then the universe may start with no particle at all at some time $t_0$ then contain a particle at some later time $t_1$ at position $x=2$. But in the present case it is difficult to say much regarding how such a space invader story is supposed to go. Unlike the classical space invader we considered earlier, here the particle would need to suddenly appear at $x=2$ without our being able to say how it got there or why it moved as it did on even the spatial portion of its trajectory.

One might seek to regain determinism by disallowing some wave functions, placing constraints on the relationship between the wave function and the particle configuration, and/or limiting the theory to a class of well-behaved Hamiltonians. But such strategies come at a cost. In each case, one is placing significant restrictions the sort of physical stories one might tell in the context of the theory.

Our approach is to use the tools of nonstandard analysis to provide a more robust characterization of the theory. This will not settle all issues regarding determinism in Bohmian mechanics, but it will allow one to represent more physical possibilities and to consider how the associated systems might evolve. Specifically, the descriptive freedom afforded by the nonstandard model will allow one to tell perfectly coherent deterministic space-invader type stories.

\section{Nonstandard analysis in a nutshell}

In this section, we briefly describe the basic premise behind nonstandard analysis.  A more detailed account can be found in Barrett and Goldbring (2021).

The starting point for nonstandard analysis is the passage to a \textbf{nonstandard extension} $\starR$ of the usual (standard) field $\bb R$ of real numbers; ${}^*\bb R$ is sometimes referred to as the field of \textbf{hyperreal numbers}.\footnote{It is important to point out that the field ${}^*\bb R$ is not in general unique.  For our purposes, it will suffice to simply pick one such field in which we work.} We assume that ${}^*\bb R$ is a proper ordered field extension of $\bb R$, and thus contains nonzero \textbf{infinitesimal elements}, that is, elements $\epsilon\not=0$ such that $|\epsilon|$ is smaller than any standard positive real number.\footnote{It follows that $\starR$ is a nonarchimedean field.}  The reciprocal of a nonzero infinitesimal number is thus an \textbf{infinite element} of $^*{\bb R}$ in that it is larger in absolute value than any standard real number.  Elements of $^*{\bb R}$ that are not infinite are (appropriately) called \textbf{finite}.  In particular, all standard real numbers are finite. 

If two elements $x$ and $y$ in $\starR$ differ by an infinitesimal (including the case that $x=y$), then we say that $x$ and $y$ are \textbf{infinitely close} and denote this by $x\approx y$.  The completeness of $\bb R$ guarantees that any finite element $x$ of $\starR$ is infinitely close to a unique standard real number, called the \textbf{standard part} of $x$, denoted $\st(x)$.

What sets the hyperreal fields apart from other nonarchimedean field extensions of $\bb R$ is that it behaves ``logically'' like $\bb R$.  More precisely, we also extend each subset $A\subseteq \bb R$ to its nonstandard extension $\starA\subseteq \starR$ in a functorial manner, meaning that any function $f:A\to B$ between subsets of $\bb R$ is also extended to a function ${}^*f:\starA\to \starB$.  Most importantly, any property that can be formulated in first-order terms is true in $\bb R$ if and only if it is true in $\starR$ when applied to the nonstandard extensions of the objects mentioned in the property; this fundamental principle, called the \textbf{transfer principle}, is at the heart of nonstandard methods.  For example, the trigonometric identity $\sin(x+y)=\sin x\cos y+\cos x\sin y$, which holds for all $x,y\in \bb R$, continues to hold for all $x,y\in \starR$, where the addition and multiplication are now those of $\starR$ and where we consider the extensions ${}^*\sin$ and ${}^*\cos$.

In order to be of real use in applications (and for our work in particular), one needs to carry out this extension process to all sets one might work with in the course of their investigations.  For us, that includes considering the \textbf{hypercomplex numbers} $\starC$ as well as the nonstandard extension $\starL^2(\bb R)$ of the usual Hilbert space $L^2(\bb R)$ of square-integrable (with respect to Lebesgue measure) functions $\bb R\to \bb C$.  By the transfer principle, elements of $\starL^2(\bb R)$ are \textbf{internal functions}\footnote{Here, an internal function is an element of the nonstandard extension of the set of functions $\bb R\to \bb C$ and, after some normalization procedure, can be viewed as an actual function $\starR\to \starC$.} $\starR\to \starC$ whose square have finite \textbf{internal Lebesgue integral}.  To make sense out of this, one can view the integral as a function from the set of nonnegative measurable functions $\bb R\to \bb R^{\geq 0}$ to the set $[0,\infty]:=\bb R^{\geq 0}\cup \{\infty\}$ which then has a nonstandard extension defined on the set of nonnegative \textbf{internally measurable}\footnote{At this point, we leave it to the reader to determine what internally measurable means.} functions $\starR\to \starR^{\geq 0}$ to $[0,\infty]^*:=\starR^{\geq 0}\cup\{\infty\}$.  By transfer, many of the standard manipulations when working with the Lebesgue integral remain valid when working with the internal Lebesgue integral.  The extension $L^2(\bb R)\subseteq \starL^2(\bb R)$ identifies an element $f:\bb R\to \bb C$ with its nonstandard extension ${}^*f:\starR\to \starC$ and the integral of $|f|^2$ equals the internal integral of ${}^*|f|^2$.  Nothing in this discussion is particular to the case of Lebesgue measure on $\bb R$ and works equally well with respect to integration on other measure spaces.

\section{A nonstandard formulation of Bohmian mechanics}

We can now characterize our nonstandard formulation of Bohmian mechanics.  The nonstandard formulation consists of the obvious internal versions of the usual posulates of Bohmian mechanics enumerated in the first section of this paper, together with one new axiom.  We begin with the former: 

\begin{itemize}
    
\item [${}^*$1.] \emph{Representation of states}: The complete physical state of a system~$S$ at time~$t$ is given by the wave function $\psi(q,t)$ over configuration space and a point in configuration space~$Q(t)$.  Now, times are hyperreal, that is, $t\in \starR$, and configuration space is $\starR^{3N}$.  Moreover, for each $t\in \starR$, $\psi(q,t)$ is now an element of $\starL^2(\bb R^{3N})$ (in particular, $q\in \starR^{3N}$), $|\psi|^2$ has internal integral equal to $1$, and the actual configuration $Q(t)$ of the particles is an element of $\starR^{3N}$.

\item [${}^*$2.] \emph{Interpretation of states}: The position of every particle is always determinate and is given by the current configuration~$Q(t)$.  Notice, nothing has changed in this postulate other than the configuration space changing from $\bb R^{3N}$ to $\starR^{3N}$.

\item [${}^*$3I.] \emph{Linear dynamics}: The wave function evolves in the standard unitary way
$$
i \hbar \frac{\partial \psi(q,t)}{\partial t} =  \hat{H} \psi(q,t)
$$
where $\hat{H}$ is the Hamiltonian.  Now, the Hamiltonian is an internal function and the above equation is to be construed as an internal partial differential equation.

\item [${}^*$3II.]  \emph{Particle dynamics}: Particles move according to
$$
\frac{d Q_k(t)}{dt} = \frac{1}{m_k} \frac{\mbox{Im }\psi^*(q,t) \nabla_k \psi(q,t)}{\psi^*(q,t)\psi(q,t)}
\Big|_{Q(t)}
$$
where~$m_k$ is the mass of particle~$k$ and~$Q(t)$ is the current particle configuration.  As in the previous postulate, the above guiding equation is now an internal ordinary differential equation.

\item [${}^*$4.]  \emph{Distribution postulate}: The probability density of the configuration~$Q(t_0)$ is $|\psi(q, t_0)|^2$ at an initial time~$t_0$.  More precisely, the probability of the configuration belonging to some internally Borel set $E\subseteq \starR^{3N}$ at an initial time $t_0$ is given by the internal integral $\int_E |\psi|^2dt:=\int_E \chi_E|\psi|^2dt$.  Note that such a probability might be infinitesimal.

\end{itemize}

New to our nonstandard formulation of Bohmian mechanics is the following postulate:

\begin{itemize}
    \item [5.]  \emph{Full wave function support}:  If the state of the system at the initial time $t_0$ is given by $\psi(q,t_0)$ and $q\in {}^*\mathbb R^{3N}$ is a \emph{physically possible} initial configuration of the system, then there is an internally open neighborhood $U$ of $q$ such that $\psi(q',t_0)\not=0$ for all $q'\in U$.
\end{itemize}

Postulate~5 is motivated by our consideration of invaders in Section~\ref{invader}. The theory allows for wave functions that provide only infinitesimal support in entire regions of configuration space, and this postulate allows for any particle configuration no matter what the wave function. In the case of the reverse space invader example, the particle will begin in a region of infinitesimal rather than zero probability. According to the distribution postulate~${}^*$4, such a state is extraordinarily improbable but not physically impossible.

In the sequel, we will make the simplifying assumption that every point in configuration space is physically possible.  The mathematics that follows can easily be adapted to deal with an arbitrary internally measurable set of physically possible configurations.

Call a (standard) Hamiltonian $\hat H$ \textbf{reasonable} if:  for any solution $\psi(q,t)$ of the Schr\"odinger equation and any initial particle configuration $Q(t_0)$ that lies in the support of the wavefunction $\psi(q,t_0)$, we have that $Q(t)$ lies in the support of $\psi(q,t)$ for all $t\in \bb R$.

It follows that if $\hat H$ is an internally reasonable Hamiltonian and the initial wave function $\psi(q,t_0)$ satisfies postulate 5 above, then $\psi(q,t)$ satisfies postulate 5 for all $t\in \starR$.

\section{Infinitesimal perturbations of standard wave functions}

In this section, we want to show how to perturb any standard wave function to an internal wave function with full support in such a way that epistemic probabilities corresponding to any standard Borel set are only changed by an infinitesimal amount.

We begin with some elementary standard measure-theoretic preliminaries.  Given a measure space $X$, we let $L^2(X)_1$ denote the unit sphere of the Hilbert space $L^2(X)$.

\begin{lem}\label{perturb}
For every $\psi\in L^2(X)_1$ and every $\epsilon>0$, there is $\tilde\psi\in L^2(X)_1$ with the following properties:
\begin{enumerate}
    \item $\|\psi-\tilde\psi\|_2<\epsilon$.
    %\item for all measurable $F\subseteq X$, $|\int_F |\psi|^2-\int_
    %F|\tilde\psi|^2|<\epsilon$,
    \item $\tilde\psi(x)\not=0$ for all $x\in X$.
\end{enumerate}
\end{lem}

\begin{proof}
Let $E:=\{x\in X \ : \ \psi(x)=0\}$, a measurable set.  Let $\theta\in L^2(X)$ be such that $\theta(x)=0$ if and only if $x\in E^c$ and such that $\delta:=\int |\theta|^2$ is sufficiently small (to be determined momentarily).  Let $\tilde\psi:=\sqrt{(1-\delta)}\psi +\theta$.  Then 
$$\int |\tilde\psi|^2=\int_E |\tilde\psi|^2+\int _{E^c}|\tilde\psi|^2=\delta+(1-\delta)=1,$$ so $\tilde\psi\in L^2(X)_1$.  It is clear from construction that $\tilde\psi(x)\not=0$ for all $x\in X$.  Note also that, for any measurable set $F$, we have $\int_F \psi^*\theta=0$. Consequently, 
$$\|\psi-\tilde\psi\|_2=\int |(1-\sqrt{1-\delta})\psi+\theta|^2=(1-\sqrt{1-\delta})^2+\delta<\epsilon$$ 
%$$\left|\int_F|\psi|^2-\int_
   % F|\tilde\psi|^2\right|\leq\int \left||\psi|^2-|\tilde\psi|^2\right|=\int (\delta|\psi|^2+|\theta|^2)=2\delta.$$
if $\delta$ is sufficiently small.
\end{proof}

% Note that there is a lot of choice for $\tilde\psi$ mainly because there is a lot of choice for $\theta$ above.  We can require, for example, $\theta$ to be a Schwarz function, whence $\tilde\psi$ would be a Schwarz function if $\psi$ is also a Schwarz function.  Also, by choosing $\theta$ to be uniformly bounded by a small number, we can even make $\psi$ and $\tilde\psi$ close pointwise.

\begin{lem}
For any functions $\psi,\tilde\psi\in L^2(X)$ and any measurable subset $F\subseteq X$, we have $|\int_F |\psi|_2-\int_F|\tilde\psi|^2|\leq 2\|\psi-\tilde\psi\|_2.$
\end{lem}

\begin{proof}
We calculate
\begin{alignat}{2}
\left|\int_F |\psi|^2-\int_F |\tilde\psi|^2\right|&=\left|\|\chi_F\psi\|_2^2-\|\chi_F\tilde\psi\|_2^2\right| \notag \\ \notag
    &\leq 2\left|\|\chi_F\psi\|_2-\|\chi_F\tilde\psi\|_2\right|\\ \notag
    &\leq 2\|\chi_F(\psi-\tilde\psi)\|_2\\ \notag
    &\leq 2\|\psi-\tilde\psi\|_2.
\end{alignat}
\end{proof}

% There is a lot of flexibility in this construction, namely in the choice of the function $\theta$.  We can make $\theta$ bounded and even a Schwarz function if that should help later on.

Now let $\psi\in L^2(\bb R^{3N})_1$ denote the wave function at some initial time $t_0$ of $N$ particles in $\bb R^3$.  Presuming that $\psi$ belongs to the domain of the Hamiltonian operator $\hat H$, we then have that there is a solution $\psi(q,t):=U(t)\psi(q)$ of the Schr\"odinger equation for which $\psi(q,t_0)=\psi(q)$, where $U(t):=e^{-it\hat H}$ is the unitary evolution operator. 

Viewing $\psi$ as a function $\psi:\bb R^{3N}\to \bb C$, we can consider its nonstandard extension $\starpsi:\starR^{3N}\to\starC$, which belongs to ${}^*L^2(\bb R^{3N})_1$.  By transfer, the nonstandard extension $\starpsi:\starR^{3N}\times \starR\to \starC$ of $\psi(q,t)$ is a solution of the internal version of the Schr\"odinger equation with ${}^*\psi(q,0)=\starpsi(q)$. 

We now apply the transferred version of Lemma \ref{perturb} to $\starpsi$ and a positive infinitesimal $\epsilon$ to obtain $\tilde{\psi}\in {}^*L^2(\bb R^{3N})_1$; we refer to $\tilde\psi$ as an \textbf{infinitesimal perturbation} of $\psi$.  By Lemma 2, given any internally Borel subset $F$ of $\starR^{3N}$, we have that the epistemic probabilities of finding the particle at time $t_0$ in $F$ with respect to the wave functions $\starpsi$ and $\tilde{\psi}$ are infinitely close.  In particular, if $F={}^*E$ for some standard Borel set $E$, we have that the epistemic probabilities of finding the particle in $E$ (with respect to $\psi$) and in ${}^*E$ (with respect to $\tilde\psi$) are infinitely close.

It is not immediately clear that there is a global solution $\tilde{\psi}(q,t)$ to the internal version of the Schr\"odinger equation for which $\tilde\psi(q,0)=\tilde\psi(q)$.  However, if the function $\theta$ appearing in the proof of Lemma 1 can be taken to belong to the domain of ${}^*\hat H$, then such a solution $\tilde{\psi}(q,t)$ does in fact exist.  For example, if the domain of $\hat H$ includes all Schwarz functions and $\psi$ itself is a Schwarz function, then this is indeed possible.

Presuming that the aforementioned global solution $\tilde{\psi}(q,t)$ does in fact exist, then by transferring the unitarity of the evolution of the wave function, item (1) of Lemma 1 remains true for $\starpsi(\cdot,t)$ and $\tilde\psi(\cdot,t)$ for all times $t\in \starR$, and thus by the transfer of Lemma 2, the epistemic probabilities of finding the particle configuration in any internally Borel subset of $\starR$ remain infinitely close for all times $t\in \starR$.  We can thus instead choose to model the given physical situation with the internal Hamiltonian ${}^*\hat H$ and the internal wavefunction $\tilde\psi$; this model now satisfies the postulates from the previous section at time $t_0$; if, in addition, $\hat H$ is reasonable, then the postulates remain valid for all time $t\in \starR$.

\section{Sanity check}

Ultimately, we want use the nonstandard formulation of Bohmian mechanics to show how one might make sense of the space invaders example. We will turn to this in the next section. However, to check that we are on the right track, we should ensure that passing from $\psi$ to $\tilde{\psi}$ as a model for our particle does not have any disastrous physical consequences for ``reasonable'' initial conditions. In particular, infinitely close epistemic probabilities does not by itself immediately entail that particle configurations remain infinitely close for all (standard) time (again, assuming reasonable initial conditions). That this desirable feature is indeed the case is the content of the following theorem.  

In this and the next section, we assume a single particle moving in one dimension.  Since configuration space is $\bb R$ in this situation, we replace the variable $q$ by the variable $x$.  

% Note also that in the above situation, we have that, for any measurable set $F$ and interpreting $\|\psi\|^2$ as a probability density, that
% $$|\bb P(x\in F)-\bb P(x'\in F)|<\epsilon.$$

% In what follows, we apply (the transfer of) Lemma 1 to $\psi(x)=\psi(x,0)\in {}^*L^2(\bb R)_1$ and some infinitesimal $\epsilon>0$ to get $\phi(x)$ (I've renamed $\psi'$ by $\phi$ to avoid confusion) which evolves to $\phi(x,t)$ with associated (internal) Bohmian motion $y(t)$.

\begin{thm}\label{maintheorem}
Suppose that $\hat H$ is a reasonable Hamiltonian and fix a function $\psi(x)\in L^2(X)_1$ and an infinitesimal perturbation $\tilde \psi(x)\in {}^*L^2(X)_1$ of $\psi$ as in the previous section.  Suppose that $x_0\in \bb R$ lies in the support of $\psi$.  Further suppose that there is a globally defined solution for the Bohmian trajectory $x(t)$ of the particle with initial wave function $\psi(x)$ and initial particle position $x(0)=x_0$ and that there is an internal globally defined solution for the internal Bohmian trajectory $\tilde{x}(t)$ of the particle with initial wave function $\tilde{\psi}(x)$ and initial particle position $\tilde{x}(0)=x_0$.  Then for every $t\in \bb R$, we have $x(t)\approx \tilde x(t)$. 
\end{thm}

% Doesn't postulate 5 require ensure that every configuration lies in the support of the wave function?

%I'm supposing that $\psi$ is a standard initial wavefunction satisfying the original Bohmian postulates.  Perhaps this should be stressed.

To prove the previous theorem, we will need the following fact about Bohmian motion in one dimension, which is an immediate consequence of the equivariance of the probability distributions and the fact that Bohmian motion in one dimension is order preserving; see equation (3) of Goldstein (1999):

\begin{fact}
For every $t\in \bb R$, we have
$$\int_{-\infty}^{x(t)}|\psi(x',t)|^2dx'=\int_{-\infty}^{x(0)}|\psi(x',0)|^2dx'.$$
\end{fact}

\begin{proof}[Proof of Theorem \ref{maintheorem}]
Let $r:=\int_{-\infty}^{x(0)}|\starpsi(x')|^2dx'$ and $s:=\int_{-\infty}^{x(0)}|\tilde\psi(x')|^2dx'$ (internal integrals).  By the transfer of Lemmas 1 and 2, we have that $r\approx s$.  Fix $t\in {}^*\bb R$.  Without loss of generality assume that $\tilde x(t)\leq x(t)$.  By the transfer of Fact 4, we have
$$r=\int_{-\infty}^{x(t)}|\starpsi(x',t)|^2dx'=\int_{-\infty}^{\tilde x(t)}|\starpsi(x',t)|^2dx'+\int_{\tilde x(t)}^{x(t)}|\starpsi(x',t)|^2dx'.$$ Now $\int_{-\infty}^{\tilde x(t)}|\starpsi(x',t)|^2dx'\approx \int_{-\infty}^{\tilde x(t)}|\tilde\psi(x',t)|^2dx'=s$, the latter equation following from the transfer of Fact 4.  It follows that $$\int_{\tilde x(t)}^{x(t)}|\starpsi(x',t)|^2dx'\approx 0.$$ If, towards a contradiction, $x(t)\not\approx \tilde x(t)$, there is some standard $z$ with $\tilde x(t)<z<x(t)$ and we have
$$\int_z^{x(t)}|\starpsi(x',t)|^2dx'\leq \int_{\tilde x(t)}^{x(t)}|\starpsi(x',t)|^2dx'\approx 0,$$ whence $\int_z^{x(t)}|\psi(x',t)|^2dx'=\int_z^{x(t)} |\starpsi(x',t)|^2dx'=0$.  Since $x(0)$ was assumed to belong to the support of the wave function, the same is true of $x(t)$ since $\hat H$ is reasonable, yielding a contradiction.
\end{proof}

\begin{remark}
The proof of Theorem 3 critically uses Fact 4, which is special to Bohmian motion in one dimension.  It would be desirable to extend the previous theorem to higher dimensions, but without something along the lines of Fact 3 holding in general, it is not clear to us how to proceed.
\end{remark}

\section{Invaders revisited}

We now revisit the space invaders example from Section~\ref{invader}. We will start with the reverse space invader case.

We considered a one-dimensional universe with an initial wave function $\psi(x)\in L^2(\bb R)$ whose support was contained in $[0,1]$ and an initial particle position $x(t_0)=2$. We now consider an infinitesimal perturbation $\tilde\psi(x)$ of $\psi(x)$ which has full wave function support in accord with Postulate~5. Under the perturbation, it is almost certain that the particle will be found in region $[0,1]$ on postulate~${}^*$4, but there is now an infinitesimal probability of finding the particle outside this region. In particular, $|\tilde\psi(x)|^2$ is a positive infinitesimal in the neighborhood of $x(t_0)$, indicating that it is extremely unlikely, but not impossible, for the initial position of the particle to be in this region. For simplicity, assuming that the motion of the particle is free and assuming that our initial wave function belongs to the domain of the usual Hamiltonian for a free particle (which is reasonable), then we can assume that our infinitesimal perturbation $\tilde\psi$ was obtained by perturbing our initial wave function $\psi$ by an infinitesimal Schwarz function and consequently conclude that there is a solution $\tilde\psi(x,t)$ to the internal version of the Schr\"odinger equation for which $\tilde\psi(x,t_0)=\tilde\psi(x)$.

Now, by transferring the usual consequences of the equations for Bohmian mechanics, we find that, after (say) one second, the particle has a well-defined position $x(t_1)$. Furthermore, as one might expect, for any $x\in \bb R$, we have that $\int_x^\infty |\tilde\psi(x',t_1)|^2dx'\approx 0$ whence, by Fact 4, it follows that $x(t_1)$ is a positive infinite element of $\starR$ and hence has a perfectly definite position represented by an infinite hyperreal number.

Note that here the exact time that one considers the position of the particle will matter as it will continue to move to the right. Further, the hyperfinite position of the particle $x(t_1)$ after one second will depend on the precise nature of the infinitesimal perturbation of the original wave function. Given concrete formulae for the infinitesimal perturbation, it is in principle possible to transfer the formulae governing the usual Bohmian evolution to predict the exact nonstandard location of the particle at any time $t_1$.

Importantly, the space invader part of the story is also now perfectly coherent. Wherever the particle ends up at time $t_1$, one can run the dynamics in reverse and obtain that the particle will indeed arrive back at $x=2$ after one second. 

In the case of the reverse space invader, we have allowed that the particle moves beyond any finite location, but we are still able to preserve our intuitions regarding the conservation of such properties as mass and charge since the particle remains at a perfectly well-defined spatial location. We are now also able to provide a coherent account of the motion of the space invader. The full story preserves both time-reversal symmetry and determinism.

The resent formulation of Bohmian mechanics allows one to represent a broader range of physical possibilities than the standard theory and to describe what would happen should those possibilities be realized. While we require that the configuration be in a region with at least infinitesimal wave function support, that the theory can handle space-invader type stories provides a concrete example of the sense in which it is both more expressive and robust than the standard formulation of the theory. In this regard, one might also note that, even with the assumption of reasonable Hamiltonians, the theory allows for unbounded energies. Its representational strength is further illustrated by the fact that it allows one to do things like consider Dirac-like distributions modeled as infinitely tall and infinitely thin rectangles and then say how a corresponding wave function might evolve in concert with the configuration.

\section{Conclusion}

The present formulation of Bohmian mechanics illustrates how one might allow for a broad range of physical states and still keep track of particle trajectories. Infinitesimal probability densities allow one to approximate any classical wave function, even those with regions of zero wave function support. And hyperreal configurations allow one to track the trajectory of particles that might otherwise move beyond any classical spacial location. As a result, one can tell perfectly deterministic space-invader type stories.

While the present theory can handle a much broader range of physical situations than the standard formulation, we are still assuming a relationship between the wavefunction and configuration and placing a constraint on physically possible Hamiltonians. Namely, Postulate 5 requires that there be at least infinitesimal wave function support for the configuration at an initial time, and if the Hamiltonian is reasonable, this preserves such support under the dynamics.

If these two conditions are satisfied, then one both captures the standard empirical predictions of Bohmian mechanics and ensures a deterministic evolution on the present formulation. We take the fact that these conditions allow one to tell space-invader type stories to provide a striking illustration of how weakly they constrain physical possibility in the context of the theory.

\newpage

\begin{center}
\large{Bibliography}
\end{center}

\vspace{.5cm}
\noindent
Barrett, Jeffrey A. (2021) ``Situated Observation in Bohmian Mechanics,” \emph{Studies in the History and Philosophy of Science}, Volume 88, August 2021, Pages 345-357.

\vspace{.5cm}
\noindent
Barrett, Jeffrey A. (2020) \emph{The Conceptual Foundations of Quantum Mechanics}, Oxford: Oxford University Press.

\vspace{.5cm}
\noindent
Barrett, Jeffrey A. (2000) ``The Persistence of Memory: Surreal Trajectories in Bohm’s Theory,'' \emph{Philosophy of Science} 67(4): 680--703.

\vspace{.5cm}
\noindent
Barrett, Jeffrey A.\ and Isaac Goldbring (2022) ``Everettian Mechanics with Hyperfinitely Many Worlds'' Published online in \emph{Erkenntnis} 10 July 2022: \\ https://doi.org/10.1007/s10670-022-00587-x.

\vspace{.5cm}
\noindent
Bell, John S.\ (1987) \emph{Speakable and Unspeakable in Quantum Mechanics}, Cambridge: Cambridge University Press.

\vspace{.5cm}
\noindent
Bohm, David (1952) ``A Suggested Interpretation of Quantum Theory in Terms of `Hidden Variables','' Parts I and II, {\em Physical Review} 85: 166--179, 180--193. In J.\ A.\ Wheeler and W.\ H.\ Zurek (eds.) (1983) \emph{Quantum Theory and Measurement}. Princeton University Press: Princeton, NJ, pp.\ 369--96.

\vspace{.5cm}
\noindent
Diaconis, Persi and Brian Skyrms (2018) \emph{Ten Great Ideas about Chance}, Princeton and Oxford: Princeton University Press.

\vspace{.5cm}
\noindent
Earman, John (1986) ``A Primer on Determinism,'' Springer: Dordrecht.

\vspace{.5cm}
\noindent
D\"urr, Detlef, Sheldon Goldstein, Nino Zangh\`i (1992) ``Quantum Equilibrium and the Origin of Absolute Uncertainty.'' \emph{Journal of Statistical Physics} 67:843--907.

\vspace{5mm}
\noindent
Ghirardi, G.\ C., Rimini, A., and Weber, T.\ (1986) ``Unified dynamics for microscopic and macroscopic systems,'' \emph{Physical Review D} 34: 470--491. 

\vspace{5mm}
\noindent
Goldstein, S. (1999) ``Absence of Chaos in Bohmian Dynamics,'' \emph{Physical Review E} 60: 7578--7579.

\vspace{.5cm}
\noindent
Wontner, N.\ J.\ H.\ (2019) \emph{Non-Standard Analysis}. 
Dissertation is submitted for the degree of Master of Philosophy, Christ’s College Faculty of Philosophy University of Cambridge 17th June 2019.

\end{document}